\documentclass[letterpaper, 10 pt, conference]{ieeeconf}  % Comment this line out if you need a4paper

\IEEEoverridecommandlockouts                              % This command is only needed if 
\usepackage{graphics} % for pdf, bitmapped graphics files
\usepackage{epsfig} % for postscript graphics files
\usepackage{mathptmx} % assumes new font selection scheme installed
\usepackage{times} % assumes new font selection scheme installed

\usepackage{amsthm}
\usepackage{amsmath}
\usepackage{amssymb}  % assumes amsmath package installed
\usepackage{amsfonts}
\usepackage{algpseudocode}
\usepackage{algorithm}
\usepackage{xcolor}
\usepackage{import}
\usepackage{cite}
\usepackage{balance}

\usepackage{graphicx}
\usepackage{float,epsf,caption,subcaption}

\newtheorem{problem}{Problem}
\newtheorem{theorem}{Theorem}

\newtheorem{corollary}{Corollary}
\newtheorem{proposition}{Proposition}
\theoremstyle{plain}

\theoremstyle{definition}

\theoremstyle{remark}

\newcommand{\R}{\mathbb{R}}

\title{\LARGE \bf
Scalable Computation of Controlled Invariant Sets for Discrete-Time Linear Systems with Input Delays
}

\author{Zexiang Liu \hspace{1cm}  Liren Yang \hspace{1cm} Necmiye Ozay% <-this % stops a space
	\thanks{The authors are with the Dept. of Electrical Engineering and Computer Science, Univ. of Michigan, Ann Arbor,
		MI 48109, USA. Emails:
		{\tt\small \{zexiang, yliren, necmiye\}@umich.edu}. This work was supported in part by NSF Awards ECCS-1553873 and CNS-1931982, and Toyota Research Institute (TRI). TRI provided funds to assist the authors with their research but this article solely reflects the opinions and conclusions of its authors and not TRI or any other Toyota entity.
		}%
}

\begin{document}

\maketitle
\thispagestyle{empty}
\pagestyle{empty}

%%%%%%%%%%%%%%%%%%%%%%%%%%%%%%%%%%%%%%%%%%%%%%%%%%%%%%%%%%%%%%%%%%%%%%%%%%%%%%%%
\begin{abstract}
In this paper, we first propose a method that can efficiently compute the maximal robust controlled invariant set for discrete-time linear systems with pure delay in input. The key to this method is to construct an auxiliary linear system (without delay) with the same state-space dimension of the original system in consideration and to relate the maximal invariant set of the auxiliary system to that of the original system. When the system is subject to disturbances, guaranteeing safety is harder for systems with input delays. Ability to incorporate any additional information about the disturbance becomes more critical in these cases. Motivated by this observation, in the second part of the paper, we generalize the proposed method to take into account additional preview information on the disturbances, while maintaining computational efficiency. Compared with the naive approach of constructing a higher dimensional system by appending the state-space with the delayed inputs and previewed disturbances, the proposed approach is demonstrated to scale much better with the increasing delay time.
\end{abstract}

%%%%%%%%%%%%%%%%%%%%%%%%%%%%%%%%%%%%%%%%%%%%%%%%%%%%%%%%%%%%%%%%%%%%%%%%%%%%%%%%
\section{Introduction}

 %!TEX root = ms.tex

As more and more autonomous functionality is introduced in human-cyber-physical systems, such as passenger vehicles and aircraft, guaranteeing their safe and correct operation becomes a major concern. From a control's perspective, safety guarantees can be obtained by computing the so called robust controlled invariant sets (RCIS) \cite{bertsekas1972infinite, kerrigan2000invariant, rungger2017computing}, which characterize the set of states from which one can find safe control actions such that the system trajectory is guaranteed to avoid the unsafe states indefinitely. However, the curse of dimensionality limits our ability to compute RCIS for high-dimensional systems.

Recently, robust controlled invariant sets are used to design controllers and supervisors for vehicle safety systems \cite{nilsson2015correct, smith}. While there exist simple linear dynamical models for vehicle control, when deploying such controllers on actual vehicles \cite{nilsson2019provably}, we have realized that there is a non-negligible time delay on the input signal, caused either by the delays in vehicle communication buses or by the dynamics of low-level actuators. Although systems with delays can be equivalently represented by higher-dimensional systems without delay by augmenting the system with additional states corresponding to delayed inputs, computing invariant sets for such high-dimensional systems is challenging. On the other hand, this equivalent augmented system is very structured and our goal in this work is to exploit this structure to compute invariant sets for systems with input delays in a scalable manner. A second concern for systems with input delays is that when the system is subject to additional disturbances, it becomes harder to guarantee invariance with a delayed input. On the other hand, future values of some of the external disturbances can be previewed by the controller \cite{peng1993preview}. In the second part of the paper, we consider the problem of incorporating such preview information in invariant set computation, without compromising scalability.

There are several results in the literature related to the invariance problem for time-delay systems. Lyapunov-based methods are recently explored by \cite{orosz2019safety, jankovic2018control}. In particular, safety barrier functionals are proposed in \cite{orosz2019safety} for general continuous-time nonlinear autonomous time-delay systems, and ``Artstein model reduction" method \cite{artstein1982linear} is used in \cite{jankovic2018control} for computing control barrier functions for continuous-time linear systems with input delay. Even though \cite{jankovic2018control} considers continuous-time systems only, it is closely related to our approach as our construction of the auxiliary system also uses a model reduction technique similar to Artstein's method but in a discrete-time setting. For methods based on discrete-time linear systems, the paper \cite{olaru2008predictive} tackles with time-varying input delay using polytope approximations, and then computes the maximal output admissible set of the closed-loop system stabilized by a linear feedback law. The papers \cite{lombardi2011positive, lombardi2012cyclic}, on the other hand, compute invariant sets for discrete-time autonomous time-delay systems with different levels of conservativeness. However, the above methods are mainly designed for systems without disturbance, have no guarantees on the maximality of the resulting controlled invariant set, and cannot incorporate preview information, thus are not applicable to our problem.

The main contributions of this paper are as follows: 
\begin{itemize}
	\item  We construct a delay-free auxiliary system by predicting the future states in $\tau$ steps (where $\tau$ is the input delay), and then show that the computation of the maximal RCIS for the high-dimensional equivalent of the delayed system can be reduced to the computation of the maximal RCIS of the low-dimensional auxiliary system and $ \tau+2 $ set intersection operations (Section \ref{sec:method}).
	\item  We extend the proposed method for systems with both input delay and disturbance preview, where the preview on disturbance can mitigate the difficulty in controlling systems with large input delays (Section \ref{sec:extension}).   
	\item We provide two examples to show the efficiency and utility of the proposed method (Section \ref{sec:example}).
\end{itemize}

\textbf{Notation:} Given two sets $S_1$ and $S_2$ in $\R^n$, the Minkowski difference between them is denoted by $S_1 \ominus S_2 = \{ x \mid x+y \in S_1,~ \forall y\in S_2 \} $. The Minkowski sum between them is denoted by $S_1 \oplus S_2 = \{x+y \mid x \in S_1, y\in S_2\} $. The Minkowski sum of a collection of sets $\{S_{i}\}_{i\in I} $ is denoted by $ \sum_{i\in I} S_{i} $. Note that we apply the convention that if $ b < a$, the Minkowski sum $\sum^{b}_{i=a} S_{i} = \emptyset $. {\color{black}Moreover, with a slight abuse of notation, we let $x+S_1$ denote the Minkowski sum of a singleton set $\{x\}$ and a set $S_1$}.  {\color{black}For a collection of sets $\{S_{i}\}_{i=1}^{n} $, the cartesian product of $S_{i}$ for $i$ from $1$ to $n$ is denoted by $S_1\times S_2 \times ... \times S_{n}$. For the case $S_1=S_2=...=S_{n}= S$, $S_1\times ...\times S_{n}$ is denoted by $S^{n}$ for short. }   For a polytope $D$ in $\R^{n}$ and a linear mapping $L: \R^{n} \rightarrow \R^{m}$, the linear transformation of $D$ with respect to $L$ is denoted by $LD \doteq \{Lx\mid x\in D \} $. 

\section{Preliminaries and Problem Statement}
\label{sec:prelim} 
In this section, we introduce some basic concepts on controlled invariant sets and then provide the problem statement.

\subsection{Maximal Controlled Invariant Set}
Consider a linear system {\color{black} $ \Sigma $ defined by}
\begin{align}
 x(t+1) = f(x(t),u(t),d(t)) \doteq Ax(t)+ B u(t) + F d(t) ,\label{eqn:sys} 
\end{align}
where $x(t)\in Q \subseteq \R^n$ is the system state, $u(t)\in U \subseteq \R^m$ is the control input and $d(t) \in D$ is a non-measurable disturbance at time $t$, and $A$, $B$, $F$ are constant matrices with appropriate dimensions. The sets $Q, U$ are the state space and input space, and $D$ is the disturbance set, respectively.  The disturbance $d$ is called \emph{non-measurable} if the input $u(t)$ is determined before the disturbance $d(t)$ is determined at each time step $t \geq 0$.

Assume that the safety constraints on $\Sigma$ is captured by a subset $X$ of $Q$, in the sense that we want the system trajectory to stay within $X$ indefinitely. We call $X$ the \emph{safe set} of $\Sigma$.  A subset $C$ of $X$ is called \emph{controlled invariant} within $X$ if for all $x\in C$, there exists $u\in U$ such that for all $d\in D$, $x^{+} = f(x,u,d)\in C$. Since by definition the union of controlled invariant sets in $X$ is a controlled invariant set, there exists a unique \emph{maximal controlled invariant set} $C_{max}$ in X that contains any controlled invariant set in $X$.

Our interest is to compute the maximal controlled invariant set $C_{max}$, because once $C_{max}$ (or any controlled invariant set) is obtained, a static feedback law can be extracted from $C_{max}$ to ensure that the trajectory of the closed-loop system stays within $C_{max} \subseteq X$ indefinitely \cite{bertsekas1972infinite}. Next we describe the fixed point algorithm for computing $C_{max}$.

Given system $ \Sigma $ as in \eqref{eqn:sys},  for a subset $V$ of $Q$, the controlled predecessor operator $Pre^{\Sigma}(V)$ is defined as 
\begin{align}
Pre^{\Sigma}(V) = \{ x \mid \exists u\in U, \forall d\in D, f(x,u,d)\in V\},
\end{align}
that is the set of states which can reach $V$ in one step robust to any disturbance $d \in D$. A set $C$ is a controlled invariant set in $X$ if and only if 
\begin{align}
	C \subseteq Pre^{\Sigma}(C)\cap X. \label{eqn:inv_cond} 
\end{align}

Define $V_0= X $ and recursively define 
\begin{align}
V_{k+1} = Pre^{\Sigma}(V_{k}) \cap X. \label{eqn:fixed_point}  
\end{align}
According to the definition, $V_{k}$ is the maximal set of states that if the system starts from $V_{k}$, there exists a feedback control law that can control the system to stay within $X$ for at least $k$ steps. Therefore, $V_{k} \supseteq C_{max}$ is an outer approximation of $ C_{max} $. 

By the monotonicity of $V_{k}$, the limit $V_{\infty} = \lim_{k \rightarrow \infty} V_{k} = \cap_{i=1}^{\infty}V_{k}$ is well defined. If $ V_k $ converges in finitely many steps, namely that there exists a $k$ such that $V_{k} = V_{k+1}$, then $ C_{max} = V_{k} $, and $C_{max}$ is called \emph{finitely determined} \cite{kerrigan2000invariant}. Otherwise, under mild conditions given in \cite{bertsekas1972infinite}, $ V_k $ converges to $ C_{max} $ in the limit, while there also exists slight modifications of this iteration with finite termination guarantees that can approximate $ C_{max} $  to arbitrary precision \cite{rungger2017computing}. Hence, in the rest of this paper, we assume that $C_{max} = V_{\infty}$ holds though we believe approximation results also carry through our analysis. 

The following proposition reveals a key property of the maximal controlled invariant {\color{black} set}, which is used to prove our main results. 
\begin{proposition}
	Suppose that $C_{max} = V_{\infty}$. Then for all $x\not\in C_{max} $, there exists $N < \infty$ such that if the system starts at $x$, the system state can be forced to reach the unsafe set $Q\backslash X$ by disturbances in at most $N$ steps. \label{prop:finite_N}
\end{proposition}
\begin{proof}
	Define $ W_k = Q \backslash V_k$ for $ k \geq 0 $ and $W_{\infty}= \cup_{i=1}^{\infty} W_{k}$. Then $W_{0}= Q\backslash X$ is the unsafe set, and as the complement of $ V_k $, $W_{k}$ is the maximal set of states that if the system starts from $W_{k}$, there exists disturbances that can force the system state to reach the unsafe set $ W_0 $ in at most $k$ steps. By De Morgan's law, $ W_{\infty} = X\backslash V_{\infty} $.
	
   Now pick $x \not\in C_{max} = V_{\infty}$. Then $x\in W_{\infty} = \cup_{i=1}^{\infty} W_{k}$, which implies that there exists a $N < \infty$ such that $x\in W_{N}$. As discussed above, if the system starts from state $ x\in W_{N}$, disturbance can force the system state to reach $W_0$ in at most $N$ steps.
\end{proof}

\subsection{Problem Statement}

	Consider a linear system with $\tau$-step pure delay in control input, that is 
\begin{align}
	\Sigma_{delay}: x(t+1) = Ax(t) + Bu(t-\tau) + F d(t) \label{eqn:sys_del} 
\end{align}
with $x(t)\in Q$, $u(t)\in U$ and non-measurable $d(t)\in D$, where the sets $Q$,  $U$ and $D$ are polytopes. 
System $\Sigma_{delay}$ can be written in the form of \eqref{eqn:sys} by appending the past $\tau$-step inputs to the state space. The resulting $(n+m\tau)$-dimensional augmented system takes the form:
\begin{align}
	\Sigma_{aug}:	\left\{ \begin{array} {rcl} 
			x(t+1)& = & Ax(t) + B u_{1}(t) + F d(t)\\
			u_{1}(t+1)& = &u_{2}(t)  \\
			u_{2}(t+1)& = & u_{3}(t)  \\
			 &\vdots &\\
	u_{\tau}(t+1)& =& u(t),\end{array}\right. \label{eqn:sys_aug} 
\end{align}
where $x(t)\in Q$, $u(t)\in U$ and non-measurable disturbance $d(t)\in D$ for all $t\geq 0$. If we want the state trajectories of the system $\Sigma_{delay}$ to remain in a polytopic safe set $X \subseteq Q$, this is the same as asking the trajectories of the augmented system $\Sigma_{aug}$ to remain in the safe set $S = X\times U^{\tau}$. Therefore, one can state the invariance problem for a system with input delay in terms of the system $\Sigma_{aug}$ as follows.

\begin{problem}
Find the maximal controlled invariant set of the system in  $\Sigma_{aug}$  within the safe set $S$. \label{prob:1}
\end{problem}

By the discussion of the preceding section, Problem \ref{prob:1} can, in principle, be solved by the iterative algorithm in \eqref{eqn:fixed_point}. However, it is well-known that this standard iterative algorithm suffers from the curse of dimensionality. Since the dimension of the augmented system increases linearly with the delay steps $\tau$, the computation of $V_{\infty}$ becomes intractable very soon as $\tau$ increases. In what follows, we propose a method to solve Problem \ref{prob:1} whose complexity is independent of $\tau$.

\section{State Prediction and Prediction Dynamics} \label{sec:method} 

Our solution approach relies on the construction of a reduced-order delay-free auxiliary dynamics with state space dimension {\color{black} the same as $x(t)$ in}  \eqref{eqn:sys_del}. We then show that the maximal controlled invariant set of the system in $\Sigma_{aug}$ can be reconstructed from the maximal controlled invariant set of the delay-free dynamics within a modified safe set. Since the dimension of new dynamics does not explode as the delay step  $\tau$ increases, the proposed method is computationally more efficient than directly computing the maximal controlled invariant set of \eqref{eqn:sys_aug}.

Our method is inspired by the following observation. Suppose that for the case $\tau=0$, the maximal controlled invariant set of $\Sigma_{delay}$ in safe set $X$ is $C$. Also, suppose that there is a sensor that can measure the future $\tau$-step values of the disturbance $d$. Then at each time $t$, the exact state evolution in $\tau$ steps, namely $x(t+1)$, ..., $x(t+\tau)$, can be calculated based on the measurement of $x(t)$, future disturbance $d(t)$, ..., $d(t+\tau-1)$ and the extended states $u_{1}(t)$, ..., $u_{\tau}(t)$, which essentially correspond to past inputs.  Then for the dynamics \eqref{eqn:sys_aug}, as long as $x(t+\tau)\in C$, we can pick $u(t)$ such that $x(t+\tau+1)\in C$. It is not hard to show that 

\begin{align}
\exists u(t) \text{ s.t. } x(t+\tau)\in X,\forall t\geq 0 \Longleftrightarrow x(\tau) \in C,\label{eqn:perfect_est} 
\end{align} 
which is very close to our goal except that $x(0),x(1),...,x(\tau-1)$ can be out of $X$. 

Moreover, by definition, $d(t), ..., d(t+\tau-1)$ are not accessible at time $t$. Nevertheless, we can predict the state evolution by assuming the future disturbances to be zero. The question is if a result similar to \eqref{eqn:perfect_est} exists. The answer is yes. We are going to show this result in the rest of this section.

First, we expand $\Sigma_{aug}$ in $\tau$ steps to obtain the exact expression of $x(t+\tau)$ as
\begin{align}
x(t+\tau)
=& A^{\tau}x(t) + \sum^{\tau}_{i=1} A^{i-1} B u_{\tau-i+1}(t) + \nonumber\\
&\sum^{\tau}_{i=1} A^{i-1} F d(t+\tau - i). \label{eqn:expand} 
\end{align}
Since $x(t)$, $u_{1}(t)$, ..., $u_{\tau}(t)$ are known at time $t$, we define  
\begin{align}
\widehat{x}_{\tau}(t) = A^{\tau}x(t) + \sum^{\tau}_{i=1} A^{i-1} B u_{\tau-i+1}(t), \label{eqn:x_wh} 
\end{align}
as a prediction of $x(t+\tau)$ based on the state measurements of dynamics $\Sigma_{aug}$ at time $t$. Define the polytope $D_{\tau}= \sum^{\tau}_{i=1} A^{i-1} F D$, which is the exact bound of the prediction error $(x(t+\tau)-\widehat{x}_{\tau}(t))$. Note that $D_{\tau} $ is time invariant and can be computed offline. Then, for all $t\geq 0$, we have the following inclusion relation:
\begin{align}
	x(t+\tau)\in \widehat{x}_{\tau}(t) + D_{\tau} \doteq \{\widehat{x}_{\tau}(t)+d \mid d\in D_{\tau}\} \label{eqn:approx_x},
\end{align}
which implies the following statement: 
\begin{align}
\widehat{x}_{\tau}(t) + D_{\tau} \subseteq X,\forall t \geq 0\Rightarrow  x(t+\tau)\in X, \forall t \geq 0. \label{eqn:equi} 
\end{align}

Therefore, if there exists a controller $u(t)$ such that $\widehat{x}_{\tau}(t)\in X\ominus  D_{\tau}$ for all $t\geq 0$, such a controller guarantees that  $x(t+\tau)\in X$ for all  $t\geq 0$. 

According to the analysis above, it is important to understand the relation between $\widehat{x}_{\tau}(t)$ and $\widehat{x}_{\tau}(t+1)$.  By definition and after some simple algebra, we have $ \Sigma_{aux} $ defined by
\begin{align}
\widehat{x}_{\tau}(t+1) & = A^{ \tau} x(t+1) + \sum^{\tau}_{i=1} A^{i-1} B u_{\tau - i+1}(t+1) \\
& = A \widehat{x}_{\tau}(t) + B u(t) +A^{\tau}F d(t),\label{eqn:sys_2}
\end{align}
where $d(t)\in D$ is a non-measurable disturbance as before. 

Thus, the problem becomes: given $\widehat{x}_{\tau}(t)\in X \ominus  D_{\tau}$, find a $u(t)\in  U$ so that for all $d(t)\in D$, $\widehat{x}_{\tau}(t+1)\in X \ominus D_{\tau}$. All that we need is to compute the maximal controlled invariant set of  $\Sigma_{aux}$  within $X \ominus D_{\tau}$, denoted as $\widehat{C}$. This computation can be done using the standard iterative algorithm in \eqref{eqn:fixed_point}. Since the dimension of $\Sigma_{aux}$  {\color{black} is equal to the dimension of $ x $} in \eqref{eqn:sys_del}, the complexity is not directly affected by the delay time $\tau$. 

Once $\widehat{C}$ is obtained, to guarantee $x(t+\tau)\in X$ for all $t\geq 0$, we need the initial state $x(0)$, $u_{1}(0)$, ..., $u_{\tau}(0)$ of dynamics $\Sigma_{aug}$ 
to be in the set

\begin{align}
C_{\tau} = \{(x(0),u_{1}(0),...,u_{\tau}(0))\mid \widehat{x}_{\tau}(0)\in \widehat{C}\}, \label{eqn:cons_1} 
\end{align}
where $\widehat{x}_{\tau}(0)$ is a function of $x(0)$, $u_{1}(0)$, ..., $u_{\tau}(0)$ defined in \eqref{eqn:x_wh}.

Furthermore, we want $x(0)$, $x(1)$, \ldots,  $x(\tau-1)$ to be within $X$. Note that $x(0)$, \ldots, $x(\tau-1)$ are determined by $x(0)$, $u_{1}(0)$, \ldots, $u_{\tau-1}(0)$, that is
\begin{align}
x(k) = A^{k}x(0) + \sum^{k}_{i=1} A^{i-1} B u_{k-i+1}(0) + \sum^{k}_{i=1} A^{i-1}Fd(k-i).
\end{align}

Therefore, for $k=0,1,...,\tau-1$, the condition under which  $x(k)$ is in $X$ for arbitrary $d(0)$, \ldots, $d(k-1)$ is
\begin{align}
	C_{k} =& \bigg\{ (x(0), u_{1}(0), \ldots, u_{\tau}(0))\bigg\vert \big( A^{k}x(0) + \sum^{k}_{i=1} A^{i-1}B u_{k-i+1}(0)  \big) \nonumber\\ 
		   &\in X \ominus  \sum^{k}_{i=1} A^{i-1}FD  \bigg\}.  \label{eqn:cons_2} 
\end{align}

Now we denote the set of states in $S = X\times U^{\tau}$ satisfying the constraints listed in \eqref{eqn:cons_1} and \eqref{eqn:cons_2} by 
\begin{align}
C_{ext} =  \left( \bigcap_{i=0}^{\tau}C_{i} \right)\cap S. \label{eqn:win_set} 
\end{align}
{\color{black}Note that each $C_k$ for $k=0,1,\ldots,\tau-1$ is just a polytope with explicit definition. Similarly, the polytope $C_{\tau}$ is also explicitly defined in \eqref{eqn:cons_1}, once the set 
$\widehat{C}$ is computed, a computation that takes place in the $n$-dimensional state space of the system $\Sigma_{aux}$. Therefore, the only polytopic operation in $(n+m\tau)$-dimensional space in computing $C_{ext}$ is the polytope intersection in \eqref{eqn:win_set}. Now, we state our main result.}

\begin{theorem}
	$C_{ext}$ is the maximal controlled invariant set of dynamics $\Sigma_{aug}$ in the set $S$. \label{thm:1}  
\end{theorem}
\begin{proof}
	Denote the maximal controlled invariant set of  $\Sigma_{aug}$ contained by S as $C_{aug}$.
	
	First, we want to show $C_{ext} \subseteq C_{aug}$. It is enough to show that $C_{ext}$ is a controlled invariant set in $S$. 
	
	Let $(x(t), u_{1}(t), u_{2}(t), ..., u_{\tau}(t))\in C_{ext}$. We want to find a $u(t)\in U$ such that for all $d(t)\in D$, $(x(t+1), u_{2}(t),...,u_{\tau}(t),u(t))\in C_{ext}$. 
	
	Since $(x(t), u_{1}(t), u_{2}(t), ..., u_{\tau}(t))\in C_{\tau}\cap S$, we have $\widehat{x}_{\tau}(t) \in \widehat{C}$. Hence, there exists $u(t)\in U$ such that for all  $d(t)\in D$, $\widehat{x}_{\tau}(t+1)\in \widehat{C}$. That is, 
	\begin{align}
	(x(t+1), u_{2}(t), ..., u_{\tau}(t), u(t)) \in C_{\tau}. \label{eqn:proof_c1} 
	\end{align}
	
	Also, since $(x(t), u_{1}(t), u_{2}(t), ..., u_{\tau}(t))\in \bigcap_{i=0}^{\tau-1}C_i$, $\{ x(t+1), ..., x(t+\tau-1)\}$ is contained by $X$. Because $x(t+1)\in X$ and $u(t)\in U$, 
	\begin{align}
	(x(t+1), u_{2}(t), ...,u_{\tau}(t), u(t))\in S. \label{eqn:proof_c2}  
	\end{align}
	
	{\color{black}Since $\widehat{x}_{\tau}(t)\in \widehat{C} \subseteq X\ominus D_\tau$, we have $x(t+\tau)\in X$ by \eqref{eqn:equi}.} Hence, for state $(x(t+1),u_{2}(t),...,u_{\tau}(t), u(t))$, for arbitrary $d(t+1),...,d(t+\tau-1)\in D$, it is verified that $\{ x(t+1), ..., x(t+\tau)\}$ is contained by $X${\color{black}, which implies}%. Hence, 
	\begin{align}
	(x(t+1),u_{2}(t),...,u_{\tau}(t),u(t))\in \bigcap_{i=0}^{\tau-1}C_{i}. \label{eqn:proof_c3} 
	\end{align}
	
	By \eqref{eqn:proof_c1}, \eqref{eqn:proof_c2} and \eqref{eqn:proof_c3}, there exists $u(t)$ such that for arbitrary $d(t)$, $$(x(t+1), u_{2}(t),...,u_{\tau}(t), u(t))\in \left( \bigcap_{i=0}^{\tau}C_{i} \right) \cap S = C_{ext}.$$ 
	Therefore, $C_{ext} $ is a controlled invariant set of dynamics $\Sigma_{aug}$  in $S$. Since $C_{aug}$ is the maximal controlled invariant set in $S$, $C_{ext} \subseteq C_{aug}$.
	
	\begin{figure}
		\centering
		\includegraphics[width=0.5\textwidth,trim=0 0 -40 -10]{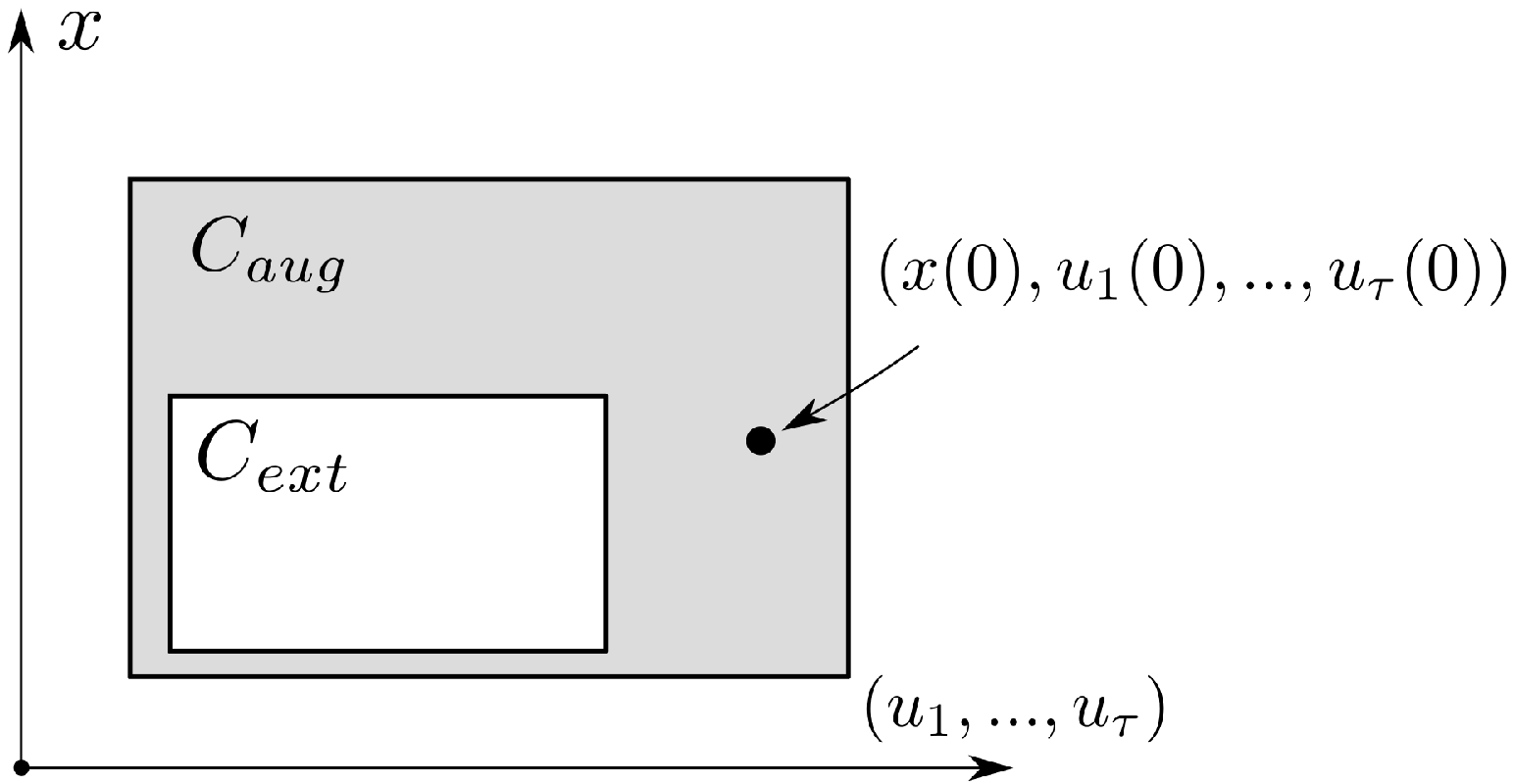}	
		\caption{Pick a point from $C_{aug}\backslash C_{ext}$.}
		\label{fig:cext_caug}
	\end{figure}
	
	Next, we want to show $C_{ext} = C_{aug}$. Suppose that there exists $$(x(0),u_{1}(0), ..., u_{\tau}(0)) \in C_{aug}\backslash C_{ext},$$ {\color{black} as demonstrated in Figure  \ref{fig:cext_caug}. }  It is easy to show that $C_{aug} \subseteq C_{k}\cap S$ for $k$ from $0$ to $\tau-1$. Since $C_{ext} = \left( \bigcap_{i=0}^{\tau-1}(C_{i}\cap S)  \right)\cap C_{\tau} \subseteq C_{aug}$, we have $(x(0),u_{1}(0), ..., u_{\tau}(0))\not\in C_{\tau}$ and thus $\widehat{x}_{\tau}(0)\not\in \widehat{C}$. Since $\widehat{C}$ is the maximal controlled invariant set of $\Sigma_{aux}$  in $X\ominus D_{\tau}$, by Proposition \ref{prop:finite_N}, there exists $N \geq 0$ such that $\forall u(0)\in U$, $\exists d(0)\in D$, $ \forall u(1)\in U$, $\exists d(1)\in D$, ..., $\forall u(N-1)\in U$, $\exists d(N-1)\in D$,  the trajectory $\{\widehat{x}_{\tau}(k)\}_{k=0}^N \not\subseteq X \ominus D_{\tau}$. That is, there exists some $ s\leq N $ such that $ \widehat{x}_\tau(s)\not\in X \ominus D_{\tau} $. Denote  $(x(s),u_{1}(s),...,u_{\tau}(s))$ as the state of  $\Sigma_{aug}$ corresponding to  $\widehat{x}_{\tau}(s)$. Followed by $\widehat{x}_{\tau}(s)\not\in X\ominus D_{\tau}$, for any possible state $(x(s),u_{1}(s),...,u_{\tau}(s))$ of  $\Sigma_{aug}$ at time $s$, there exists $d(s), d(s+1),..., d(s+\tau-1)\in D$ such that $x(s+\tau)\not\in X$, which holds for arbitrary control inputs $u(s),...,u(s+\tau-1)$. Contradiction to the assumption that $C_{aug}$ is a controlled invariant set of dynamics $\Sigma_{aug}$. Therefore, $C_{ext}=C_{aug}$.
\end{proof}

{\color{black} 

The following corollary directly follows from the first part of the proof of Theorem \ref{thm:1}.
\begin{corollary}
	If $\widehat{C}$ is a controlled invariant set of $\Sigma_{aux}$ in $X \ominus  D_{\tau}$ (but not necessarily the maximal one), then $C_{ext}$ is a controlled invariant set of $ \Sigma_{aug}$ in $S$.
\end{corollary}
The above corollary shows that our method can be applied even if $\widehat{C}$ is not the maximal controlled invariant set of $\Sigma_{aux}$, which is very useful because typically it is easier to find a controlled invariant set than the maximal controlled invariant set.
} 

\section{Extension to Systems with Preview} \label{sec:extension} 
In the previous section, the disturbance is assumed to be completely non-measurable, but in real-world systems, many external signals can be previewed in ahead of time by sensors (see, for instance, the lane keeping example in Section~\ref{sec:LK_ex} where the disturbance term is the lane curvature). In this section, we study the maximal controlled invariant set for time-delayed systems whose disturbance inputs can be separated into a non-measurable disturbance and a disturbance with preview. 

More specifically, a disturbance is called \emph{with $k$-step preview} if $d(t)$, $d(t+1)$, ..., $d(t+k-1)$ can be measured by the controller at time $t$. Hence, the control input $u(t)$ can be picked depending on $d(t), d(t+1), ..., d(t+k-1)$.  By definition, a non-measurable disturbance is with $0$-step preview.

Consider the following linear system
\begin{align}
	\Sigma_{delay}^{prev}: x(t+1) = Ax(t) + Bu(t-\tau) + F_{0} d_{0}(t)+  F_{p} d_{p}(t), \label{eqn:sys_del_pvw} 
\end{align}
where $x(t)\in Q \subseteq \R^n$, $u(t)\in U\subseteq \R^m$, $d_{0}$ is a non-measurable disturbance in $D_{0}$ and $d_{p}$ is a disturbance with $p$-step preview in $D_{p}\subseteq \R^l$ ($1\leq  p \leq \tau$)\footnote{Note that the results in this section can be generalized for systems that have disturbances with different preview lengths (less than or equal to $\tau$). 
For the sake of simplicity, we restrict the discussion to the case where the system has disturbance only with a fixed preview length.}. The state space $Q$, input space $U$, disturbance bounds $D_0$ and $D_{p}$ are polytopes.  In the analysis that follows, we assume  $p \leq \tau$ and show that in this case one can still compute the invariant set by applying the invariance iterations in $n$-dimensional space. While the case when $p > \tau$ can be handled with some added complexity within the current framework, whether this case can also be reduced to iterations in an $n$-dimensional space is left for future research.   

Similar to the delay, a system with preview can also be converted to a standard linear system by appending the state space with addition states corresponding to the preview of the disturbance. In particular, we have the following augmented system equivalent to $\Sigma_{delay}^{prev}$:

\begin{align}
\Sigma_{aug}^{prev}:\left\{
		\begin{array}{rcl}
		x(t+1) &=& Ax(t) + B u_{1}(t) + F_{0} d_{0}(t)+ F_{p}d_{p,1}(t)\\
		u_{1}(t+1) &=& u_{2}(t)  \\
		&\vdots&\\
		u_{\tau}(t+1) &=& u(t)\\
		d_{p,1}(t+1) &=& d_{p,2}(t)\\
		d_{p,2}(t+1) &=& d_{p,3}(t)\\
		&\vdots&\\
		d_{p,p}(t+1) &=& d_{p,f}(t),
	\end{array} \right. \label{eqn:sys_aug2} 
\end{align}
where $x(t)\in Q$, $u(t)\in U$, $d_{0}(t)$ and $d_{p,f}(t) $ are non-measurable disturbances bounded by $D_{0}$ and $D_{p}$. Note that $d_{p,f}(t)$ is just an alias of $d_{p}(t+p)$. The safe set of system $\Sigma_{aug}^{prev}$  is $S_{p}= X\times U^{\tau}\times D_{p}^{p}$.

\begin{problem}
    Find the maximal controlled invariant set of system $\Sigma_{aug}^{prev}$  within the safe set $S_{p}$.
\end{problem}

Proceeding as in the previous section, we can write the $\tau$-step expansion of $x(t+\tau)$ as
\begin{align}
	x(t+\tau)
	=& A^{\tau} x(t)+  \sum^{\tau}_{j=1} A^{j-1}Bu_{\tau-j+1}(t) + \nonumber\\
	 &\sum^{\tau}_{j=\tau-p+1}  A^{j-1}F_{p}d_{p,\tau-j+1}(t) +\nonumber\\
	 &\sum^{\tau}_{j=1}  A^{j-1}F_{0}d_{0}(t+\tau-j)+ \nonumber\\
	 &\sum^{\tau-p}_{j=1}  A^{j-1}F_{p}d_{p,f}(t+\tau-p-j).
\end{align}
Based on what can be measured at time $t$, we define the prediction variable 
\begin{align}
	\widehat{x}_{\tau}(t) =& A^{\tau} x(t)+  \sum^{\tau}_{i=1} A^{i-1}Bu_{\tau-i+1}(t) +\sum^{\tau}_{j=\tau-p+1}  A^{j-1}F_{p}d_{p,\tau-j+1}(t) \label{eqn:x_pred2} 
\end{align}
assuming the non-measurable disturbance is zero. 
The dynamics of $\widehat{x}_{\tau}$ takes the form
\begin{align}
	\Sigma_{aux}^{prev}: \widehat{x}_{\tau}(t+1) =& A \widehat{x}_{\tau}(t) + Bu(t) + A^{\tau}F_{0}d_{0}(t)+\nonumber\\
	&A^{\tau-p}F_{p}d_{p,f}(t), \label{eqn:sys_3} 
\end{align}
where $u(t)\in U$, $d_{0}\in D_{0}$ and $d_{p,f}(t)\in D_{p}$ are non-measurable disturbances, which is again an $n$-dimensional system. 
From Eq.~\eqref{eqn:x_pred2}, it follows that to ensure $x(t+\tau)\in X$, we need 
\begin{align}
	\widehat{x}_{\tau}(t) \in \widehat{X} := X \ominus \sum_{i\in\{0,p\} } \sum^{\tau-i}_{j=1}  A^{j-1}F_{i}D_{i}.
\end{align}

Let $\widehat{C}_{p}$ as the maximal controlled invariant set of  $\Sigma_{aux}^{prev} $ within $\widehat{X}$. We next show how to construct an invariant set for the $(n+m\tau+pl)$-dimensional augmented system $\Sigma_{aug}^{prev}$ using the set $\widehat{C}_{p}\subseteq \R^n$. The constraints on the initial states of the system $\Sigma_{aug}^{prev}$, ensuring the existence of a controller such that $x(t+\tau)\in X $ for all $t \geq 0$, can be written as

\begin{align}
	C_{p,\tau} = \{(x(0),u_{1}(0),...,d_{\tau,\tau}(t)) \mid \widehat{x}_{\tau}(0)\in \widehat{C}_{p}\},
\end{align}
where $\widehat{x}_{\tau}(0)$ is defined by \eqref{eqn:x_pred2}. 

To ensure $x(k)\in X$ for the time period $k = 0,\ldots,\tau-1$, we need the following constraints 
on initial states:
{\small
\begin{align} 
	C_{p,k}=& \bigg\{ (x(0),u_{1}(0),\ldots,d_{\tau,\tau}(t)) \mid A^{k}x(0)+ \sum^{k}_{j=1} A^{j-1}Bu_{k-j+1}(0) +\nonumber\\ 
			&\sum^{k}_{j=\max\{1,k-p+1\} } A^{j-1}F_{p}d_{p,k-j+1}(0) \nonumber\\
			&\in X \ominus \sum_{i\in\{0,p\} } \sum^{k-i}_{j=1} A^{j-1}F_{i}D_{i} \bigg\}
\end{align}
}

Finally, define the intersection of these constraint sets:
\begin{align}
	C_{p,ext} = (\bigcap_{k=0}^{\tau} C_{p,k})\cap S_{p}.
\end{align}

\begin{theorem}
$C_{p,ext}$ is the maximal controlled invariant set of system $\Sigma_{aug}^{prev}$  in set $S_{p}$. \label{thm:2} 
\end{theorem}
\begin{proof}
The proof for Theorem \ref{thm:2} can be easily extended from the proof of Theorem \ref{thm:1}, omitted for brevity. 
\end{proof}

{\color{black} Similar to the preceding section, we have the following corollary.
\begin{corollary}
	If $\widehat{C}_{p}$ is a controlled invariant set of  $\Sigma_{aux}^{prev}$ in $X \ominus \sum_{i\in\{0,p\} } \sum^{\tau-i}_{j=1} A^{j-1} F_{i}D_{i}  $, $C_{p,ext}$ is a controlled invariant set of $\Sigma_{aug}^{prev}$ in $S_{p}$.
\end{corollary}
} 

\section{Examples} \label{sec:example} 
 %!TEX root = ms.tex
 
 The algorithms are implemented in MATLAB 2018b on a computer equipped with Intel i7-8650U CPU and $16$ GB memory. We use implementations from MPT3 toolbox \cite{MPT3} for the polytope operations in the algorithms.

\subsection{Numerical Example}
In this section, we use a toy example to show how much performance improvement is achieved by applying the proposed method. 

Consider the following $1$-dimensional system:

\begin{align}
x(t+1) = 1.5x(t) + u(t-\tau) + d(t)
\end{align}
where $x(t) \in \R$, $u(t)\in [-20,20]$ and $d(t)\in [-2,2] $ with $p$-step preview. The safe set for $x$ is taken to be $[-32,32]$.

In Table \ref{tab:time_cmp}, we compare the computation time of the maximal invariant sets for different $\tau$ and $p$ using two methods: the proposed method and the fixed-point algorithm in  \eqref{eqn:fixed_point} operating on the augmented system. We call the later the direct method for short. {\color{black} We note that the iterations \eqref{eqn:fixed_point} terminates in finite number of steps in all of the examples in the table.} There are two important observations. 

First, for each  $\tau$ in Table \ref{tab:time_cmp}, $p$ is selected as the smallest preview length that makes the maximal invariant set nonempty. The increasing trend on $p$ in Table  \ref{tab:time_cmp} implies that if we do not have any preview on disturbance, the controlled invariant set becomes empty very soon as $\tau$ increases. That reveals how preview on disturbance reduces conservativeness for input-delay systems, which is why we take preview into consideration in Section \ref{sec:extension}. 

Second, according to the last two columns of Table \ref{tab:time_cmp}, the computation time with the direct method increases drastically as $\tau$ increases, while the computation time for the proposed method just increases slightly. This is because the dimension of the reduced-order system does not change as $\tau$ and $p$ increase. Our method is apparently more efficient than the direct method in this example.

\begin{table}
	\centering
	~\\
	\caption{Time required to compute an invariant set with the proposed method and the direct method. }
	\label{tab:time_cmp}
	\begin{tabular}{cccc}
\hline
$\tau$   &  $p$  & proposed method ($s$) &  direct method ($s$) \\
\hline
$1$  & $0$ & $0.7705$ & $0.5960$ \\ 
$5$  & $ 1 $ & $ 0.8779 $ & $ 7.7573 $\\
$ 10 $ & $ 6 $ & $1.1548 $ &$  98.3379 $\\
$15$ & $11$  & $ 1.6999$ &  $   525.7656$\\ 
$20$ & $16$ & $ 3.0460 $ & $ 1.6217\times 10^3 $\\
\hline
	\end{tabular}
\end{table}

\subsection{Vehicle Lane Keeping Control}\label{sec:LK_ex}
In this example, the proposed method is applied to synthesize a controller that guarantees the safety of a vehicle in a lane-keeping scenario. The goal of lane keeping is to control the vehicle to follow the center line of the road. The safety requirement is to make sure the lateral displacement, the lateral velocity, yaw angle and yaw rate of the vehicle with respect to the road center are within given bounds so that the vehicle does not leave the target road, spin, or rollover. 

The vehicle dynamics considered is linearized from a bicycle model \cite{smith} and discretized by forward Euler method with time step $h=0.1s$. The longitudinal velocity  $v_{d}$ is fixed and equal to $30m/s$.  The state of the system consists of the lateral displacement $y$ between the vehicle center and the road center, the lateral velocity $v$, the yaw angle $\Delta \Psi$ and the yaw rate $r$ of the vehicle, denoted by $x = [y,v,\Delta \Psi, r] $. The dynamics $ \Sigma_{car} $ of $x$ is 
\begin{align}
	x(t+1)= (I+A\cdot h)x(t)+ Bh \delta_f(t-\tau) + Fh r_{d}(t) \label{eqn:vehicle} 
\end{align}
with $I$ equal to the identity matrix and 
\begin{align*}\footnotesize
    A = \begin{bmatrix}
    0  & 1  & u  & 0\\
	0  & - \frac{C_{\alpha f}+C_{\alpha r}}{mu}  & 0  & \frac{bC_{\alpha r}- a C_{\alpha f}}{mu}-u \\
	0  &  0 & 0 & 1\\
	0  & \frac{bC_{\alpha r}- a C_{ \alpha f}}{I_{z} u}  & 0  & - \frac{a^{2}C_{ \alpha f}+ b^{2}C_{ \alpha r}}{I_{z}u} 
    \end{bmatrix}, B= \begin{bmatrix}
        0\\
		\frac{C_{\alpha f}}{m}\\
		0 \\
		a \frac{C_{\alpha f}}{I_{z}}
	\end{bmatrix}, F = \begin{bmatrix}
        0 \\
        0 \\
        -1 \\
        0	
	\end{bmatrix}, 
\end{align*}
where the steering angle $\delta_{f}\in [-\pi/2,\pi/2]$ is the control input with $\tau$-step delay and the desired yaw rate $r_{d}$ is a disturbance with $p$-step preview ($p\leq \tau$).   The parameters in $A$, $B$ matrices are taken from \cite{smith}. According to \cite{michiganroad}, the maximal range of $r_{d} $  with respect to $v_{d}= 30m/s$ in Michigan is $ D = [-0.05,0.05]$. Desired yaw rate $r_{d}$ is a function of the road curvature, which can be measured with a forward looking camera or acquired from a map ahead of time. Therefore it is reasonable to assume that $r_{d}$ is a disturbance with preview. 

The safe region $X$ of dynamics \eqref{eqn:vehicle} is given by bounds $ \vert y \vert \leq 0.9 $, $ \vert v \vert \leq 1.2 $, $ \vert \Delta \Psi \vert \leq 0.05 $ and $ \vert r \vert  \leq 0.3$. For $\tau = 10, p = 8$, our method takes $249 s$ to compute the maximal controlled invariant set of the $22$-dimensional augmented system within the safe set $X \times U^{10} \times D^{8}$. By fixing $r $, $ u_1 $, $\ldots$, $u_{10} $, $d_1 $, $\ldots$, $d_{8}$ to be zero, we make a $3$-dimensional slice of the $22$-dimensional polytope, shown in Figure~\ref{fig:slice}. The red region in Figure~\ref{fig:slice} contains all the feasible initial values of the first three coordinates $( y, v, \Delta \Psi)$ from which it is possible to guarantee safety, when the other coordinates have initial value equal to $0$.

\begin{figure}[h]
	\centering
	\includegraphics[width=0.5\textwidth]{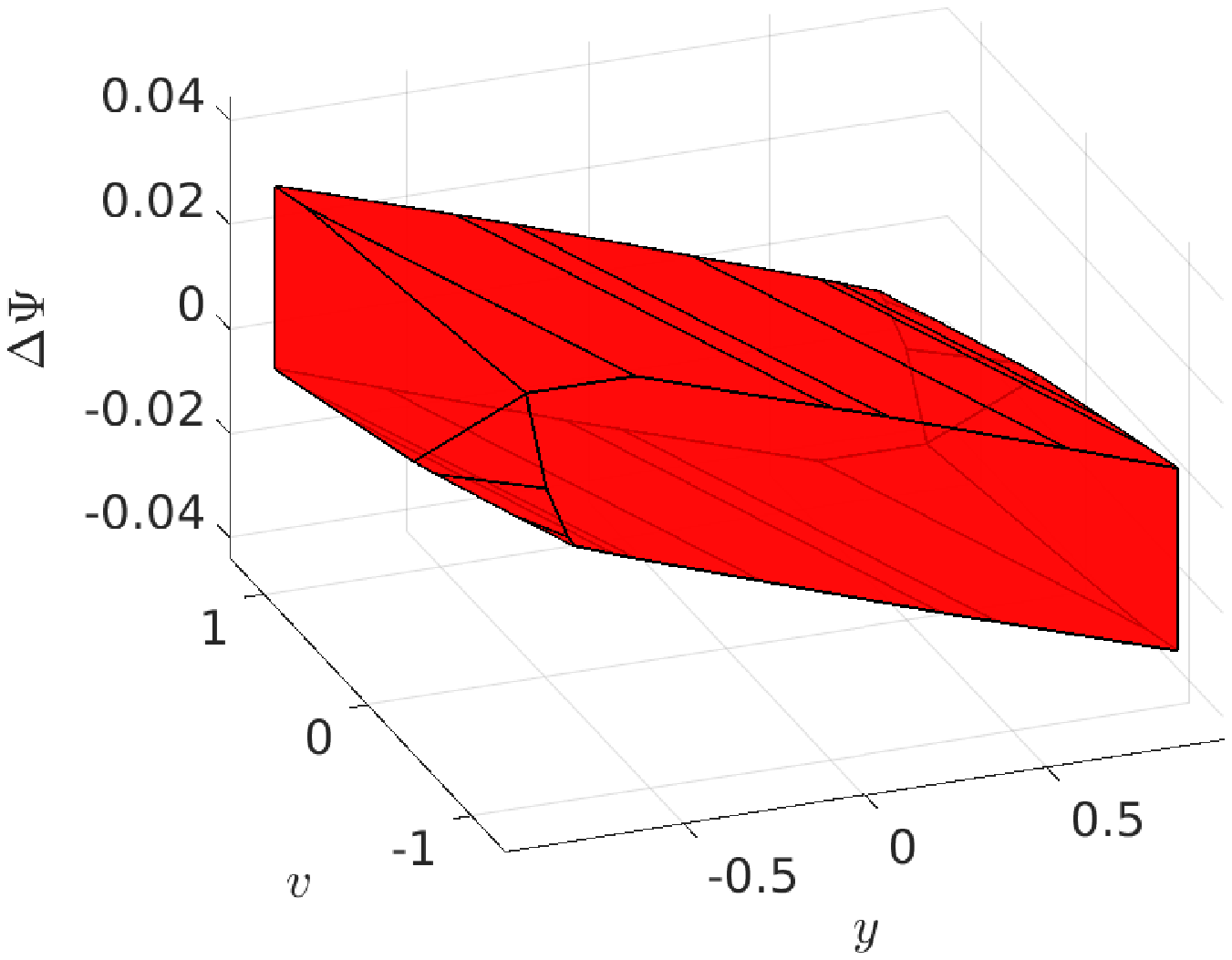}
	\caption{A slice of the maximal controlled invariant set.}
	\label{fig:slice}
\end{figure}

Once the maximal controlled invariant set $C$ is obtained, the admissible input set with respect to a state in $C$ is the set of inputs that make the next state within $C$ robust to any disturbances in $D$. A safety supervisor for a legacy vehicle controller or human-driver can be implemented by checking if the controller's output is within the admissible input set at each time and making appropriate adjustment \cite{nilsson2015correct}.

{\color{black} We run a simulation under the supervisory control framework using the maximal controlled invariant set of $ \Sigma_{car} $. In the simulation, $r_{d}$ is given by a sine function over time. A legacy controller $ u_n $ of the vehicle is obtained by solving a Linear Quadratic Regulator problem for the augmented system of $ \Sigma_{car} $. The supervisor is implemented by projecting the output of $ u_n $ to the admissible input set given by the maximal RCIS of $ \Sigma_{car} $ at each time step. As a baseline, we first assume that the invariant set designer is either unaware of the existence of the delay and preview or simply ignores them and implements the supervisor using the maximal controlled invariant of $ \Sigma_{car} $ with zero delay and no preview. Then, another supervisor is implemented based on the maximal controlled invariant set of $ \Sigma_{car} $ with the actual delay steps and preview steps. A sample trajectory of the closed-loop systems equipped with the first and second supervisors are compared in Figure \ref{fig:simu}, indicated by red and blue curves. The red trajectory terminates at $ 4.9 $s because at that time the system equipped with the first supervisor reaches the unsafe region. In contrast, the system equipped with the second supervisor stays within the safety bounds all the time. Comparing the two different simulation results, it can be seen that simply ignoring the delay can lead to unsafe situations. It is also worth noting that the invariant set becomes empty in this example when taking the preview time $p$ to be zero while keeping the delay time as is. In fact, for any value of $p<8$, the invariant set is empty. This indicates the value of preview in coping with uncertainty for systems with input delays. }

\begin{figure}
	\centering
	\includegraphics[width=0.4\textwidth,trim=0 0 0 -10]{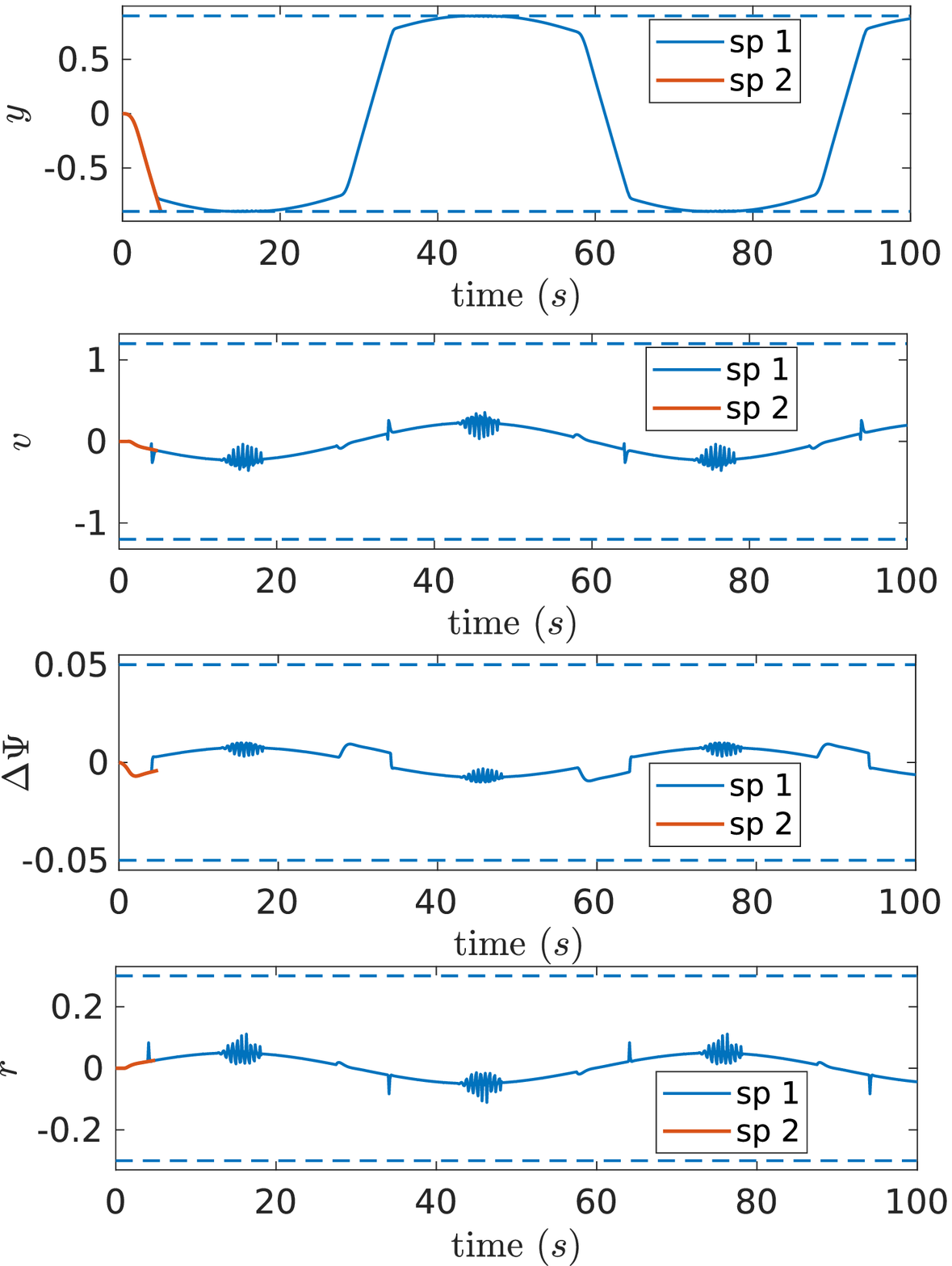}
	\caption{Trajectories of the supervisory control simulation. The safety bound on each coordinates are indicated by the dash lines. The red and blue trajectories correspond to supervisors designed with different knowledge on the delay time.}
	\label{fig:simu}
\end{figure}

\section{Conclusions}
In this paper we propose a scalable method for computing controlled invariant sets for linear systems subject to input delays. This method is extended to incorporate preview information while preserving the scalability properties. Both of the problems studied are motivated by safety control problems in automotive domain, yet we believe the proposed methods are broadly applicable. Our current work focuses on understanding the robustness of the approach to uncertainties in the delay time. We are also interested in time-varying delays where the correctness and maximality guarantees will depend on the protocol that resolves missing or clashing input packets.

\section*{Acknowledgment}
The authors would like to thank Kevin Zaseck from TRI for valuable discussions motivating this work.
\balance
\bibliographystyle{IEEEtran}
\bibliography{acc}

\end{document}